\documentclass[a4paper,10pt]{article}
\usepackage[T1]{fontenc}
\usepackage[english]{babel}
\usepackage{amssymb,amsmath,latexsym,epsfig}

\usepackage[latin1]{inputenc}
\usepackage[T1]{fontenc}
\usepackage[all]{xy}
\usepackage{algpseudocode}

\makeatletter
\newcommand\ackname{Acknowledgements}
\if@titlepage
  \newenvironment{acknowledgements}{%
      \titlepage
      \null\vfil
      \@beginparpenalty\@lowpenalty
      \begin{center}%
        \bfseries \ackname
        \@endparpenalty\@M
      \end{center}}%
     {\par\vfil\null\endtitlepage}
\else
  \newenvironment{acknowledgements}{%
      \if@twocolumn
        \section*{\abstractname}%
      \else
        \small
        \begin{center}%
          {\bfseries \ackname\vspace{-.5em}\vspace{\z@}}%
        \end{center}%
        \quotation
      \fi}
      {\if@twocolumn\else\endquotation\fi}
\fi
\makeatother

\title{Generalized Hamming weights for almost affine codes\footnote{The original publication is available at http://ieeexplore.ieee.org/document/7820189/}}
\author{Trygve Johnsen\thanks{ Dept. of Mathematics, UiT The Arctic University of Norway, N-9037 Troms{\o}, Norway, \texttt{Trygve.Johnsen@uit.no}}  \and Hugues Verdure\thanks{ Dept. of Mathematics, UiT The Arctic University of Norway, N-9037 Troms{\o}, Norway, \texttt{Hugues.Verdure@uit.no}} }

\usepackage{bm}

\newtheorem{definition}{Definition}
\newtheorem{proposition}{Proposition}
\newtheorem{corollary}{Corollary}
\newtheorem{theorem}{Theorem}
\newtheorem{example}{Example}

\newtheorem{remark}{Remark}
\newtheorem{lemma}{Lemma}
\newenvironment{proof}[1][Proof]{\begin{trivlist}
\item[\hskip \labelsep {\bfseries #1}]}{\end{trivlist}}
\newcommand{\rk}{\textrm{rk}}
\newcommand{\Fq}{\mathbb{F}_q}
\newcommand{\C}{\mathcal{C}}
\newcommand{\N}{\mathbb{N}}

\newcommand{\qed}{\nobreak \ifvmode \relax \else
      \ifdim\lastskip<1.5em \hskip-\lastskip
      \hskip1.5em plus0em minus0.5em \fi \nobreak
      \vrule height0.75em width0.5em depth0.25em\fi}

\begin{document}

\maketitle

\begin{abstract}\noindent 
We define generalized Hamming weights for almost affine codes.  We show that this definition is natural since we can extend some well known properties of the generalized Hamming weights for linear codes, to almost affine codes. In addition we discuss duality of almost affine codes, and of the smaller class of multilinear codes.

\end{abstract}

\section{Introduction}

Let $C$ be an almost affine code as defined in \cite{SA},  that is: $C \subset F^n$ for some finite alphabet $F$, and the projection $C_X$ has cardinality $|F|^s$ for a non-negative integer $s$ for each $X \subset \{1,\cdots,n\}.$

It is well known (\cite{SA}) that $C$ defines a matroid $M_C$ through the rank function \[r(X)=\log_{|F|} |C_X|.\] Such codes were studied in connection with access structures over $E=\{1,2,\cdots,n\}$ and are strongly related to ideal perfect secret sharing schemes for such access structures. See e.g.~\cite{SA},~\cite{JM},~\cite{BD},~\cite{Ma}. Recently, almost affine codes have been used in network coding. See e.g.~\cite{WEH}

An important subclass of almost affine codes are linear codes over finite fields $\mathbb{F}_q$.  A bigger class consists of affine codes, which are translates of linear codes within their ambient space. Another class of codes strictly contained in the class of all almost affine codes, consists of multilinear codes (see Section~\ref{five} for the definition of multilinear codes).

In this paper we will study some well-known properties of linear codes over finite fields, and investigate to what extend they carry over to this bigger class of almost affine codes $C$.

We start by defining the Hamming weights of almost affine codes and show that the different characterizations of Hamming weights for linear codes apply to almost affine codes.

We carry on by investigating the possibility of defining in a natural way  a dual code $C^{\perp}$ of an almost affine code $C$. This turns out to be problematic in general, although the dual matroid of $M_C$ exists, so that we know what matroid structure $C^{\perp}$ should have induced, if it  had existed. For multilinear codes, however, there is a nice duality of codes, which matches that of the dual matroids.

We proceed to prove a version of Kung's theorem for almost affine codes, that is a formula for how many codewords it takes for their unions of supports to cover all of $E=\{1,2,\cdots,n\}.$ For linear codes this formula in formulated in terms of the minimum distance  of the dual code. In our case there is not necessarily a dual code, but we succeed in formulating a similar result, by using the associated matroid of the code. We also extend a recent generalization of Kung's theorem, given in \cite{JSV}, from linear codes to almost affine codes. Here we give formulas for how many codewords it takes for their unions of supports to cover subsets of $E=\{1,2,\cdots,n\}$ of specified cardinalities. To formulate this result we use the full set of Hamming weights for the matroid $M_C$.

At the end of the the paper, we look at two notions from linear codes that transpose nicely to almost affine codes, and that emphasize that our definition of Hamming weights is the right one. Namely, we look at dimension/length profiles of an almost affine code and its application to trellis decoding, and at the wire-tap channel of type II. In both cases, the Hamming weights of the code give an indication on how complex decoding will be, and how much information an intruder can get.

Our exposition contains several examples of almost affine codes that are not necessarily linear. Apart from a simple running example introduced in Example~\ref{running} below, we look at codes arising from a simple interleaving scheme (Section~\ref{five}) and folded Reed-Solomon codes (Section~\ref{VI}). The way almost affine codes
arise in a natural way from ideal perfect secret sharing schemes is also explained, in Section~\ref{IB}

\subsection{Notation and known results}

\subsubsection{Matroids} A matroid is a combinatorial structure that extend the notion of dependency. There are many equivalent definitions for matroids, but we give just one here. We refer to~\cite{O} for a complete overview of the theory of matroids, and we use the notation from~\cite{O}.

\begin{definition}
A matroid $M$ is a pair $(E,\rho)$ where $E$ is a finite set and $\rho : 2^E \rightarrow \N$ is a function satisfying \begin{description} \item[(R1)] $\rho(\emptyset)=0,$\item[(R2)] If $X \subset E$ and $x \in E$, then \[\rho(X) \leqslant \rho(X \cup \{x\}) \leqslant \rho(X)+1.\]\item[(R3)] If $X \subset E$, $x,y \in E$ and \[\rho(X)=\rho(X \cup \{x\}) =\rho(X \cup \{y\})\] then \[\rho(X \cup \{x,y\})=\rho(X).\] \end{description}
A basis of the matroid is a subset $X \subset E$ such that $|X|=\rho(X)=\rho(E)$, while a circuit is a minimum subset of $X \subset E$ (for inclusion) satisfying $\rho(X)=|X|-1$. The nullity function is the function \[n(X)=|X|-\rho(X).\] The rank of the matroid is $\rho(E)$.
\end{definition}

\begin{remark}If $C$ is a $[n,k]$ linear code over a finite field $\Fq$, we can associate to it a matroid $M(C)$ in the following way: let $H$ be a parity check matrix of the code. Then $E= \{1,\cdots,n\}$ and the rank function is given by \[\rho(X) = \rk_{\Fq} H_X\] for $X \subset E$, where $H_X$ is the submatrix of $H$ obtained by keeping the columns indexed by $X$. It can be proved that this matroid does not depend on the parity check matrix.\end{remark}

Every matroid $M$ admits a dual matroid $M^*$ on the same ground set and with rank function \[\rho^*(X) = |X|+\rho(E \- X) - \rho(E).\] Of course, $\left(M^*\right)^*=M$.

A notion that will be used later is the fundamental circuit of an element with respect to a basis~\cite[Corollary 1.2.6]{O}: 
\begin{definition}\label{fundamental}
If $B$ is a basis and $e \in E\-B$, then there exists a unique circuit $X$ such that $X \subset B \cup \{e\}$. This circuit will be denoted $\sigma(B,e)$ in the sequel.
\end{definition}

In~\cite[Theorem 2]{W}, Wei generalizes the notion of minimum distance of linear codes (the generalized Hamming weights), and this can be further extended to matroids in general (\cite{JV}):

\begin{definition} Let $M$ be a matroid of rank $k$ on the ground set $E$, and let $n$ be its nullity function. Then the generalized Hamming weights are \[d_i(M)=\min\{|X|, n(X)=i\} \textrm{ for } 1 \leqslant i \leqslant |E|-k.\]
\end{definition}

Notice that the generalized Hamming weights for a matroid are a strictly increasing function of $i$.

In the same way, we can define the generalized Hamming weights for the dual matroid $M^*$. These are related by Wei duality, first proved in~\cite[Theorem 3]{W} for linear codes, and then generalized in~\cite{L} (in Norwegian) and also in~\cite[Theorem 5]{BJMS}, where one may disregard the partial ordering  $P$ appearing in that theorem since we now are considering the case where $P$ is trivial (antichain):

\begin{proposition}\label{weis}

The  $d_i(M)$ and the  $d_i(M^*)$ satisfy Wei duality: \[\{d_1(M),\cdots,d_{n-k}(M)\} \textrm{ } \cup \] 
$$\{n+1-d_{k}(M^*),\cdots,n+1-d_1(M^*)\}=\{1,2,\cdots,n\}$$ where $n=|E|$.
\end{proposition}

\subsubsection{Almost affine codes}\label{IB} We refer to~\cite{SA} for an introduction to almost affine codes, and will mainly use their notation. We give here the main definitions, and the result that will be used in the sequel. 

\begin{definition}An almost affine code on a finite alphabet $F$, of length $n$ and dimension $k$ is a subset $C \subset F^n$ such that $|C|=|F|^k$ and such that for every subset $X \subset \{1,\cdots,n\}$, \[\log_{|F|}|C_X| \in \N,\] where $C_X$ is the puncturing of $C$ with respect to $\{1,\cdots,n\}\-X$.

The code $C$ is non-degenerate when it is of effective length $n$, that is, when $\forall x \in \{1,\cdots,n\}$, $\log_{|F|}|C_{\{x\}}| > 0.$

An almost affine subcode of $C$ is a subset $D \subset C$ which is itself an almost affine code on the same alphabet.
\end{definition}

To any almost affine code $C$ of length $n$ and dimension $k$ on the alphabet $F$, we can associate a matroid $M_C$ on the ground set $E=\{1,\cdots,n\}$ and with rank function \[r(X) = \log_{|F|} |C_X|,\] for $X \subset E$.

It is easily checked that this is the rank function of a matroid. The first axiom is trivial. The second axiom comes from the fact that a new coordinate position either leaves the number of codewords unchanged, or increases it by a factor $|F|$. The third axiom comes from the fact that if the number of codewords do not increase when we add new coordinate positions $x$ or $y$, then it does not increase when we add both.

\begin{remark} Obviously, any linear code $C$ over the field $\Fq$ is an almost affine code on the alphabet $\Fq$. We have two matroids associated to this code, namely $M(C)$ and $M_C$. Unfortunately, they are different, but they remain related, since they are dual of each other. We have namely \[M_C = M(C)^* = M(C^\perp)\] where $C^\perp$ is the dual linear code of $C$, that is the orthogonal complement of $C$.
\end{remark}

\begin{example}\label{running}We will use a running example throughout this paper. It is the almost affine code $C'$ in~\cite[Example 5]{SA}. It is a code of length $3$ and dimension $2$ on the alphabet $F=\{0,1,2,3\}$. Its set of codewords is \begin{align*} 000 && 011 && 022 && 033\\101&&112&&123&&130 \\ 202&&213&&220&&231\\303&&310&&321&&332\end{align*} Its matroid is the uniform matroid $U_{2,3}$ of rank $2$ on $3$ elements. Namely, $r(\{1,2,3\}) = \log_416 = 2$ while for any $X \subsetneq \{1,2,3\}$, it is is easy to see that $C'_X=F^{|X|}$ so that $r(X)=|X|.$ This is an example of an almost affine code which is not equivalent to a linear code, and not even to a multilinear code. \end{example}

When talking about the support of a codeword in a linear code, one implicitly makes reference to the zero-codeword. Such a "canonical" codeword does not generally exist in almost affine codes, so we are bound to specify the codeword we compare to in almost all our definitions.

\begin{definition} Let $C$ be a block code of length $n$, and let $\bm{\tilde{c}} \in C$ be fixed. The $\bm{\tilde{c}}$-support of any codeword $\bm{c}$ is \[Supp(\bm{c},\bm{\tilde{c}}) = \{i, \bm{c}_i \neq \bm{\tilde{c}}_i\}.\]
\end{definition}

Even if this is defined using a fixed codeword $\bm{\tilde{c}}$, it is shown in~\cite{SA}, that many quantities defined for almost affine codes do not depend on the codeword $\bm{\tilde{c}}$ used, but just on the matroid associated to the code. We mention, among other definitions and results taken from~\cite{SA}:

\begin{definition} \label{fixed}
Let $C$ be an almost affine code of length $n$, and let $\bm{\tilde{c}} \in F^n$ be fixed. Then \[C(X,\bm{\tilde{c}}) = \{\bm{c} \in C,\ \bm{c}_X = \bm{\tilde{c}}_X\},\] where $\bm{c}_X$ is the projection of $\bm{c}$ to $X$.
\end{definition} 

\begin{proposition}\label{prop2} Let $C$ be an almost affine code of length $n$ and dimension $k$ on the alphabet $F$. Let $\bm{\tilde{c}} \in C$. Let $X \subset \{1,\cdots,n\}$. Then $C(X,\bm{\tilde{c}})$ is an almost affine subcode of $C$, and moreover, \[|C(X,\bm{\tilde{c}})| = |F|^{k-r(X)}\] where $r$ is the rank function of the matroid $M_C$.
\end{proposition}

\begin{corollary}\label{surbase}
If $B$ is a basis of $M_C$, then given any tuple $\bm{w} \in F^B$, there exists a unique word $\bm{w}' \in C$ such that $\bm{w}'|_B = \bm{w}$.
\end{corollary}

\begin{proof}
Such a word exists since by definition of a basis, $C_B = F^B$, and it is unique by the previous proposition, since $r(B)=k$.
\end{proof}

In the sequel, some proofs can be made clearer if one uses a equivalent code instead. Two block codes $C$ and $C'$ of length $n$ on alphabets $F$ and $F'$ respectively are equivalent if there exists a permutation $\sigma \in S_n$ and bijections $\tau_i:F \rightarrow F'$ for $1 \leqslant i \leqslant n$ such that $C'$ is the result of applying $\tau_i$ to the symbols in position $i$ for all words in $C$, for $1 \leqslant i \leqslant n$, followed by permuting the $n$ digits of each word according to $\sigma$.

It is obvious that a code equivalent to an almost affine code is almost affine too. It will be obvious in the sequel that it will be enough to prove the properties we want to prove for an equivalent almost affine code. Then we can assume that the alphabet is $F=\{0,\cdots,q-1\}$, that $\{1,\cdots,k\}$ is a basis of the matroid associated to the code, and that the word $(0,\cdots,0) \in C$.

\subsection{The relation with access structures and ideal perfect sharing schemes}

The interest in almost affine codes has arisen in a natural way in connection with 
secret sharing schemes and their associated access structures. The connection with these structures is thoroughly explained for 
example in \cite{SA}, and we briefly recollect some central elements, to motivate our study of almost affine codes. We essentially follow the exposition in \cite{SA}.

Let $E_1=\{2,3,\cdots,n\}$ be a set of $n-1$ participants, for an integer  $n \ge 2$.

\begin{definition}$ $
\begin{itemize}
\item An access structure over $E_1$ is a set $\Gamma$ of subsets of $E_1$, such that $A \in \Gamma$ and $A \subset B$ implies $B \in \Gamma$.
\item For an access structure $\Gamma$ we let  $\Gamma_0$ denote  the set of minimal elements of $\Gamma$.
\item The access structure $\Gamma$ is said to be connected if the union of the sets in $\Gamma_0$ is all of $E_1$.
\end{itemize}
\end{definition}

Let $F$ be a finite set of secrets, and denote by $q$ its cardinality. A perfect secret sharing scheme for the access structure
$\Gamma$  is a method of distributing shares to the participants in such a way that all groups of
participants in $\Gamma$ can retrieve the secret, but no other group has any a posteriori information
about the secret. A perfect secret sharing scheme is said to be ideal if the share set for each
participant is equal to the set of secrets $F$. In mathematical terms:

\begin{definition}
Set $E= \{1,2,\cdots, n\}$, and denote by $\overline{A}$ the set $A \cup \{1\}$, for any $A \subset E_1.$
An ideal perfect secret sharing scheme for the access structure $\Gamma$ is a subset $\mathcal{C} \subset F^E (=F^n))$
such that:
\begin{itemize}
\item $C_{\{i\}}=F$, for $i=1,\cdots,n.$ 
\item $|\mathcal{C}_{\overline{A}}| = |\mathcal{C}_{A}|,$ for all $A \in \Gamma.$
\item  $|\mathcal{C}_{\overline{A}}| = q |\mathcal{C}_{A}|,$ for all $A$ not contained in  $\Gamma.$
\end{itemize}
\end{definition}

It is then clear that if you start with a non-degenerate almost affine code $\mathcal{C} \subset F^n$, then $\mathcal{C}$
 is a  ideal secret sharing scheme for the access structure $\Gamma_{\mathcal{C}}$ defined by

$(\Gamma_{\mathcal{C}})_0 = \{A \subset E_1 | \overline{A}$ is a circuit in $M_{\mathcal{C}}\}$

\begin{definition}
A matroid with ground set $E$ is connected if every subset of $E$ of cardinality $2$ is contained in  a circuit.
\end{definition}
It is then clear that the access structure $\Gamma_{\mathcal{C}}$  is connected if and only if the matroid $M_\mathcal{C}$ 
 is connected.

We also have (\cite{BD}):

\begin{proposition} 
An ideal perfect secret sharing scheme for a connected access structure is an almost affine code.
\end{proposition}
For more on this subject we refer to \cite{SA}, \cite{JM}, \cite{BD}, \cite{Ma}.

\section{Generalized Hamming weights}

\subsection{Definition via the associated matroid}

For a block code $C$, let $d(\bm{x},\bm{y})$ be the Hamming distance between the codewords $\bm{x}$ and $\bm{y}$, that is $d(\bm{x},\bm{y}) = |Supp(\bm{x},\bm{y})|.$ The minimal distance $d$ is defined as \[d= \min\{d(\bm{x},\bm{y}),\ \bm{x},\bm{y} \in C,\ \bm{x} \neq \bm{y}\}.\] 
Then from~\cite[Prop. 5]{SA}, the minimal distance of an almost affine code $C$ is equal to the minimum cardinality of the circuits of the dual of the matroid associated to $C$, in other words, \[d= d_1(M_C^*).\] This suggests the following definition of generalized Hamming weights for an almost affine code:

\begin{definition}
The generalized Hamming weights for an almost affine code $C$ of dimension $k$ are \[d_i(C)=d_i(M_C^*) = \min\{|X|,\ |X| - r^*(X) = i\}\]  for  $1\leqslant i \leqslant k$, where $r^*$ is the rank function of $M_C^*$.
\end{definition} 

\begin{example} Let $C'$ be the almost affine code of Example~\ref{running}. The dual of $M_{C'=}U_{2,3}$ is $M_{C'}^*=U_{1,3}$, the uniform matroid of rank $1$ on $3$ elements. Its generalized Hamming weights are \begin{align*}d_1(C')&= d_1(M_{C'}^*)=2 \\ d_2(C')&=d_2(M_{C'}^*)=3.\end{align*} 
\end{example}

\begin{proposition} \label{fix}
Let $C$ be an almost affine code of length $n$ and dimension $k$ on the alphabet $F$. Let $\bm{\tilde{c}} \in C$ be any codeword. Then for every $1 \leqslant i \leqslant k$, \begin{eqnarray*} d_i(C) &=& \min\{|X|,\ r(E\-X) = k-i\}\\ & =& n - \max \{ |X|,\ r(X) = k-i\} \\ &=& n- \max\{|X|,\ |C(X,\bm{\tilde{c}})|=|F|^i\}.\end{eqnarray*} The third equality is independent of the choice of $\bm{\tilde{c}}$.
\end{proposition}
\begin{proof}
 The first equality follows simply from the fact that \[r^*(X) = |X| + r(E\-X) + k\] while the third equality is derived from Proposition~\ref{prop2}. 
\end{proof}

\subsection{Generalized Hamming weights and subcodes}

For linear codes, the generalized Hamming weights are originally defined as minimal supports of linear subcodes of a given dimension (\cite{W}). While for linear codes of dimension $k$ over the finite field $\Fq$, the number of linear subcodes of dimension $1 \leqslant i \leqslant k$ is known, namely $\genfrac{[}{]}{0pt}{}ki_q$, this is not the case for almost affine codes. Even two almost affine codes having the same associated matroid do not necessarily have the same number of almost affine subcodes. Nevertheless, we can express the generalized Hamming weights for an almost affine code in terms of supports of almost affine subcodes.

\begin{definition}
Let $C$ be an almost affine code, and let $\bm{\tilde{c}} \in C$. The $\bm{\tilde{c}}$-support of $C$ is \[Supp(C,\bm{\tilde{c}}) = \bigcup_{\bm{w} \in C} Supp(\bm{w},\bm{\tilde{c}}).\]
\end{definition}

\begin{lemma} Let $C$ be an almost affine code, and $\bm{\tilde{c}},\bm{\tilde{d}} \in C$. Then we have \[Supp(C,\bm{\tilde{c}}) = Supp(C,\bm{\tilde{d}}).\]
\end{lemma}

\begin{proof} Namely, let $i \in \bigcup_{\bm{w} \in C} Supp(\bm{w},\bm{\tilde{c}})$. Then there exists $\bm{w} \in C$ such that $\bm{w}_i \neq \bm{\tilde{c}}_i$. If $\bm{w}_i \neq \bm{\tilde{d}}_i$, then of course $i \in \bigcup_{\bm{w} \in C} Supp(\bm{w},\bm{\tilde{d}})$. Otherwise $\bm{\tilde{c}}_i \neq \bm{w}_i = \bm{\tilde{d}}_i$ and again, $i \in \bigcup_{\bm{w} \in C} Supp(\bm{w},\bm{\tilde{d}})$. By symmetry, we get equality.
\end{proof}

The support of any almost affine subcode is thus well defined, as long as we take the $\bm{\tilde{c}}$-support of any codeword $\bm{\tilde{c}}$ in the subcode, and we may omit the reference to this codeword. For linear codes, we have an obvious candidate that is in any subcode, namely the $\bm{0}$-codeword. For almost affine codes, we may have to use different codewords for different subcodes. Indeed, in the almost affine code $C'$ of Example~\ref{running}, the following almost affine subcodes of dimension $1$ are disjoint: \[\{0,0,0\},\{1,0,1\},\{2,0,2\},\{3,0,3\}\] and \[\{1,1,2\},\{2,1,3\},\{0,1,1\},\{3,1,0\}.\] In that case, their supports are $(1,3)$ for both.

\begin{theorem}\label{th1}
Let $C$ be an almost affine code of length $n$ and dimension $k$ on an alphabet $F$ of cardinality $q$. Then the generalized Hamming weights for $C$ are \[d_i(C) = \min \left\{\begin{array}{c}|Supp(D)|,\ D\textrm{ is an almost affine} \\  \textrm{subcode of dimension }i\textrm{ of }C\end{array}\right\}\] for $1 \leqslant i \leqslant k$.
\end{theorem}

\begin{remark} Almost affine subcodes of dimension $i$ always exist by Proposition~\ref{prop2}, since we can always find in the matroid $M_C$ a set $X$ with $r(X)=k-i$.
\end{remark}

\begin{proof}[Proof of Theorem~\ref{th1}]
For $1 \leqslant i \leqslant k$, let \[d_i=d_i(C)\]  and \[e_i = \min\left\{\begin{array}{c}|Supp(D)|,\ D\textrm{ is an almost affine subcode}\\\textrm{ of dimension }i\textrm{ of }C\end{array}\right\}.\] We show first that $d_i \leqslant e_i$. Let $D$ be an almost affine subcode of $C$ of dimension $i$ such that $|Supp(D)|=e_i$. By definition of the dimension, $|D| = q^i$. Let $\bm{\tilde{d}} \in D \subset C$, and let $X= Supp(D,\bm{\tilde{d}})$. We look at $D' = C(E\-X,\bm{\tilde{d}})$. By Proposition~\ref{prop2}, we know that this is an almost affine subcode of dimension $l= k- r(E\-X).$ It is obvious that $D \subset D'$, and in particular \[i \leqslant l = k- r(E\-X).\] By the monotone property of generalized Hamming weights for matroids, we have that \[d_i \leqslant d_l = \min \{|Y|, k-r(E\-Y)=l\} \leqslant |X| = e_i.\]

We show now that $e_i \leqslant d_i$. Let $X \subset E$ be such that $|X| = d_i$ and $r(E\-X) = k-i$. Consider $D''=C(E\-X,\bm{\tilde{c}})$ where $\bm{\tilde{c}}$ is any codeword of $C$. By Proposition~\ref{prop2}, the dimension of $D''$ is $i$. Of course $\bm{\tilde{c}} \in D''$, and by construction $Supp(D'',\bm{\tilde{c}}) \subset X$. Then 
\begin{eqnarray*}d_i = |X| &\geqslant& |Supp(D'',\bm{\tilde{c}})|\\& \geqslant& \min \{|Supp(D), \dim\ D=i\} = e_i.\end{eqnarray*}
\end{proof}

\begin{example} Let $C'$ be the almost affine code of Example~\ref{running}. This code has $12$ almost affine subcodes of dimension $1$, and it can be shown that all of them have support of cardinality $2$. One of these subcodes is $\{022,332,202,112\}$ which has support $\{1,2\}$.
\end{example}

\subsection{Generalized Hamming weights and codewords}

In~\cite{JV}, it is shown that the nullity function (and a posteriori the generalized Hamming weights) can be expressed as the support of non-redundant circuits. 

\begin{definition} 
Let $\{X_1,\cdots,X_s\}$ be a set of distinct subsets of a given set.
We say that this is a non-redundant set of subsets if the union of the $s$ subsets is not equal to any union of $s-1$ of the subsets.
\end{definition}

By abuse of notation we then also just say that $X_1,\cdots,X_s$ are non-redundant subsets.

From~\cite{JV} we have:
\begin{proposition} \label{supp}
Let $M$ be a matroid and $X$ a subset of the ground set. Then the nullity of $X$ is equal to the number of elements in a maximal non-redundant subset of circuits included in $X$.
\end{proposition}

For linear codes, circuits of the  matroid associated to (any) parity check matrix are in one to one correspondence with supports of minimal codewords. In~\cite[Proposition 5]{SA}, it is proved that an analogous result holds for almost affine codes, namely that if $C$ is an almost affine code and $\bm{\tilde{c}} \in C$, then the $\bm{\tilde{c}}$-supports of the $\bm{\tilde{c}}$-minimal codewords are the circuits of the dual matroid associated to the code. They are of course independent of the codeword $\bm{\tilde{c}}$. 
This gives rise to the following:
\begin{definition}
Let $\bm{\tilde{c}}$ be a codeword in an almost affine code $C$. A set $\{\bm{c_1},\cdots,\bm{c_i}\} \subset C$ is called a $\bm{\tilde{c}}$-non-redundant set of codewords if  $\{Supp(\bm{c_1},\bm{\tilde{c}}),\cdots,Supp(\bm{c_i},\bm{\tilde{c}})\}$ is a non-redundant set of subsets. It is called a $\bm{\tilde{c}}$-minimal non-redundant set of codewords if in addition the $\bm{c_j}$ are $\bm{\tilde{c}}$-minimal for all $j$.
\end{definition}

By abuse of notation we also just say that $\bm{c_1},\cdots,\bm{c_i}$
are $\bm{\tilde{c}}$-non-redundant codewords (respectively $\bm{\tilde{c}}$-minimal non-redundant codewords), and we may omit the reference to $\bm{\tilde{c}}$ when there is no risk of confusion.

Proposition~\ref{supp} gives rise to the following characterization of the generalized Hamming weights for a matroid.

\begin{proposition}Let $M$ be a matroid of rank $k$ on the ground set $E$. Then the $i$-th generalized Hamming weight, for $1 \leqslant i \leqslant |E|-k$ is given by \[d_i(M) = \min\left\{|\bigcup_{j=1}^i X_j|,\begin{array}{c}\ X_1,\cdots,X_i \\ \textrm{ are non-redundant circuits}\end{array}\right\}.\]
\end{proposition}

\begin{proof} Let \[d_i=\min\{|X|,\ n(X)=i\}\] and \[e_i=\min\{|\bigcup_{j=1}^i X_j|,\ X_1,\cdots,X_i \textrm{ are non-redundant circuits}\}\]
Let $X_1\cdots,X_i$ non-redundant circuits such that $|\bigcup X_j|=e_i$, and let $Y=\bigcup X_j$. Then by Proposition~\ref{supp}, $j=n(Y) \geqslant i$. By the monotony of the generalized Hamming weights for a matroid, \[d_i \leqslant d_j \leqslant |Y| = e_i\] and one inequality is proved. For the second inequality, let $Y \subset E$ such that $|Y|=d_i$ and $n(Y)=i$. Then by Proposition~\ref{supp} again, there exists $i$ non-redundant circuits $Y_1,\cdots,Y_i$ such that $\bigcup Y_j \subset Y.$ Then \[e_i \leqslant |\bigcup Y_j| \leqslant |Y| = d_i\] and this proves the proposition.
\end{proof}

Then we have the following characterization of the generalized Hamming weights for an almost affine code (and thus linear code):
\begin{proposition}\label{proposition6} Let $C$ be an almost affine code of dimension $k$. Then the generalized Hamming weights for $C$ are given by \[d_i(C)=\min \left\{\begin{array}{c}|\bigcup_{j=1}^i Supp(\bm{c_j},\bm{\tilde{c}})|,\ (\bm{c_1},\cdots,\bm{c_i})\textrm{ are }\\ \bm{\tilde{c}}- \textrm{minimal non-redundant codewords}\end{array}\right\}\]
\end{proposition}

For a linear code, we have that a linear subcode of dimension $i$ and minimal support gives $i$ codewords with non-redundant supports that define $d_i$, and the converse. And actually, that any $i$ non-redundant codewords defines a linear subcode of dimension $i$. This is not the case for almost affine codes. There is for example no almost affine subcodes of dimension $1$ in the code $C'$ of Example~\ref{running} containing the origin (in this case $000$) and the word $112$.

\begin{lemma} Let $D \subset C$ be an almost affine subcode of dimension $i$ and such that $|Supp(D)|=d_i(C)$. Let $\bm{\tilde{c}} \in D$. Then we can find $\bm{c_1},\cdots,\bm{c_i}\in D$, $\bm{\tilde{c}}$ non-redundant and such that \[|\bigcup_{j=1}^i Supp(\bm{c_j},\bm{\tilde{c}})| = |Supp(D)| = d_i(C).\]
\end{lemma}

\begin{proof} Without loss of generality, we may assume that $F = \{0,\cdots,|F|-1\}$ and that $\bm{\tilde{c}}$ is the $0$ word. Let $X$ be a basis of $M_D$. In particular, by Corollary~\ref{surbase}, there exists for each $x \in X$ a (unique) word $\bm{c_x} \in D$ such that $(\bm{c_x})_{X\-\{x\}} = (0,\cdots,0)$ and $(\bm{c_x})_x=1$. Let $\bm{d_x}\in D$ be a word such that $Supp(\bm{d_x},\bm{\tilde{c}})$ is minimal and contained in $Supp(\bm{c_x},\bm{\tilde{c}})$. We claim that $x \in Supp(\bm{d_x},\bm{\tilde{c}})$. Namely, if not,  then \[Supp(\bm{{d}}_x,\bm{\tilde{c}}) \subset Supp(\bm{c_x},\bm{\tilde{c}}) \subset E \- (X\-\{x\})\] together with $x \not \in Supp(\bm{d_x},\bm{\tilde{c}})$ would imply that $(\bm{d_x})_X=(0,\cdots,0)$, that is, $\bm{d_x}=\bm{\tilde{c}}$ by Corollary~\ref{surbase} again, which is absurd. Thus, these codewords $\bm{d_x}$ are $\bm{\tilde{c}}$-minimal non-redundant. Then by Proposition~\ref{proposition6}, we have that \[|\bigcup_{x \in X} Supp(\bm{c_x},\bm{\tilde{c}})| \geqslant|\bigcup_{x \in X} Supp(\bm{d_x},\bm{\tilde{c}})| \geqslant d_i(C).\] 
By construction, since all the $\bm{c_x} \in D$, \[\bigcup_{x \in X} Supp(\bm{c_x},\bm{\tilde{c}}) \subset Supp(D)\] so that \[d_i(C) \leqslant |\bigcup_{x \in X} Supp(\bm{c_x},\bm{\tilde{c}})| \leqslant |Supp(D)| = d_i(C)\] and there must be equality everywhere.
\end{proof}

And the converse:

\begin{lemma} Let $C$ be an almost affine code and $\bm{\tilde{c}} \in C$. Assume that $\bm{c_1},\cdots,\bm{c_i}$ are $\bm{\tilde{c}}$-minimal non-redundant and such that $|\bigcup Supp(\bm{c_j},\bm{\tilde{c}})| = d_i(C)$. Then there exists an almost affine subcode $D$ of $C$ containing $\bm{\tilde{c}},\bm{c_1},\cdots,\bm{c_i}$, of dimension $i$, and $|Supp(D)|=d_i(C)$. 
\end{lemma}

\begin{proof} 
Let $X=\bigcup_{j =1}^i Supp (\bm{c_j},\bm{\tilde{c}})$.  We have that \[|X|=d_i(C) < d_l(C) = \min\{|Y|,\ n^*(Y)=l\}\] for every $i<l$, so that $n^*(X) \leqslant i$. The inequality $n^*(X) \geqslant i$ is a direct consequence of~\cite[Proposition 5]{SA} and Proposition~\ref{supp}. This shows that  $n^*(X) = i$, i.e $i=k-r(E\-X).$ This also means that the subcode $D=C(E\-X,\bm{\tilde{c}})$ is an almost affine subcode of dimension $i$ by Proposition~\ref{prop2}. By construction, $\bm{c_i} \in D$ for all $i$, and of course $\bm{\tilde{c}} \in D$. Moreover, generally, $Supp(C(E\-X,\bm{\tilde{c}})) \subset X$, so that $|Supp(D)| \leqslant |X| = d_i(C)$.  By Theorem~\ref{th1}, there has to be equality.\\
\end{proof}

\section{Duality and Wei duality} 

For linear codes, we can easily define a dual code, namely the orthogonal complement of the code. The generalized Hamming weights for the code and its dual are related by Wei duality (\cite[Theorem 3]{W}). This was generalized to matroids (coming from linear codes or not), as presented in Proposition~\ref{weis}. So, if $C$ is an almost affine code, we could define the dual generalized Hamming weights as the generalized Hamming weights for the dual of the associated matroid, and we would get a Wei duality by Proposition~\ref{weis}, coming essentially from matroid theory. It would be nice if these weights would come from a dual almost affine code. Unfortunately, we will see that such duals do not exist in general. But for a large class of almost affine codes, we can nevertheless define a dual code.

\subsection{The dual of an almost affine code does not exist in general}

It is natural to ask the following about dual almost affine codes:
 \begin{itemize}
\item The matroid associated to the dual code should be the dual of the matroid associated to the code.
\item Two equivalent codes should have equivalent duals.
\item The dual of the dual should be the code we started with. 
\end{itemize} 
In addition, the dual of a linear code and of a linear code seen as an almost affine code should coincide.
\begin{remark} In the case of linear codes, we replace the condition on equivalent codes by a stronger condition, namely linear equivalence. It is unknown to the authors if  two linear codes can be equivalent in the wider sense without being linearly equivalent.
\end{remark}

\begin{lemma}\label{lem1} Let $C_1,C_2$ be two equivalent almost affine codes on the alphabet $F$. Then for every $1 \leqslant r \leqslant \dim(C_1)=\dim(C_2)$, the number of $r$-dimensional almost affine subcodes of $C_1$ and $C_2$ are the same.
\end{lemma}

\begin{proof}This is obvious by the definition of equivalency.\end{proof}

\begin{lemma}\label{lem2} Let $C_1,C_2$ be two almost affine codes of dimension $1$ on the alphabet $F$ with the same matroid. Then they are equivalent.
\end{lemma}

\begin{proof} Let $B=\{b\}$ be a basis of the matroid. Let $x \in E\-B$. We have two possibilities: \begin{itemize} \item $\sigma(B,x)= \{x\}$.  Let $\bm{w_i} \in C_i$ for $i \in \{1,2\}$. By Proposition~\ref{prop2}, \[\left|C_i(\{x\},\bm{w_i})\right| = |F|^{1-r(\{x\})} = |F| = |C_i|\] so that all words of $C_i$ have the same digit, namely $(\bm{w_i})_x$ at position $x$. Let $\tau_x$ be any permutation of $F$ that sends $(\bm{w_1})_x$ to $(\bm{w_2})_x$.
\item $\sigma(B,x)=\{b,x\}$.  
For every $i \in \{1,2\}$ and $f \in F$, let $\bm{w_{i,f}} \in C_i$ be the unique word such that $(\bm{w_{i,f}})_{\{b\}}=f$. Since $\{x\}$ is a basis of $M_{C_1}=M_{C_2}$, by Corollary~\ref{surbase}, \[\begin{array}{cccc} \tau_x: &F &\longrightarrow& F \\ &\left(\bm{w_{1,f}}\right)_x & \longmapsto & \left(\bm{w_{2,f}}\right)_x \end{array}\] is a permutation.
\end{itemize}The series of permutations $\tau_x$ of the symbols of the alphabet at position $x$ makes $C_1$ equivalent to $C_2$.
\end{proof}

We can now show that the concept of dual of an almost affine code does not exist. Namely, the codes $\C$ and $\C'$ from~\cite[Example 5]{SA}, have the same associated matroid. The dual matroid is the uniform matroid $U_{1,3}$. Therefore, the possible duals $\C^\perp$ and $\C'^\perp$ would be equivalent by Lemma~\ref{lem2}. Thus $\C={\C^\perp}^\perp$ and $\C'={\C'^\perp}^\perp$ would also be equivalent. But this is not possible by Lemma~\ref{lem1} since it is known that $\C$ has $20$ $1$-dimensional almost affine subcodes, while $\C'$ has just $12$ of them.

\subsection{Duality of multilinear codes} \label{five}

In this subsection, we will study an important class of almost affine codes, namely multilinear codes. 

\begin{definition} 
Let $q$ be a prime power and $r,n \geqslant 1$. Let $F$ be the $\Fq$-vector space $\Fq^r$. A multilinear code $C$ is a $\Fq$-linear subspace of $F^n$ such that $\forall X \subset \{1,\cdots,n\}$, $\dim_{\Fq} C_X$ is divisible by $r$.
\end{definition}

\begin{example}
Let $C$ be a $[n,k]$ linear code on the field $\Fq$ with generator matrix $G=\begin{bmatrix}g_{i,j}\end{bmatrix}$. Let $F$ be the $\Fq$-vector space $\Fq^r$ for some $r$. Consider the following interleaving encoding scheme: \[\xymatrix{m_1\ar[d] & \cdots\ar[d] & m_k\ar[d] \\
m_{11} & \cdots & m_{1k} \ar[r]^{\cdot G}& c_{11} & \cdots & c_{1n} \\
\vdots & \ddots & \vdots \ar[r]^{\cdot G}& \vdots & \ddots & \vdots \\
m_{r1} & \cdots & m_{rk} \ar[r]^{\cdot G}& c_{r1}\ar[d] & \cdots\ar[d] & c_{rn}\ar[d] \\
&&&c_1 & \cdots &c_n}\] where $m_i \in F$ is decomposed into $m_{1,i},\cdots,m_{r,i} \in \Fq$. Then every row $[m_{j,1},\cdots,m_{j,k}]$ is encoded via $G$ to a row $[c_{j,1},\cdots,c_{j,n}]$. Now all the columns $c_{l,1} \cdots,c_{l,r}$ forms an element of $F$. This code $C'$ is the row space of the $kr \times rn$ block matrix \[G'=\begin{bmatrix} D_{g_{ij}}^{(r)} \end{bmatrix}\] on $\Fq$, where $D_l^{(r)}$ is the $r \times r$ diagonal matrix with $l$ on the diagonal. This is therefore a multilinear code.
\end{example}

It is shown in~\cite{SA} that a multilinear code $C$ is an almost affine code on the alphabet $F=\Fq^m$. The rank function of the associated matroid is given by \[\rho(X) = \frac{1}{r} \dim_{\Fq}C_X,\ X \subset \{1,\cdots,n\}.\]

By the canonical isomorphism $F^n \approx \Fq^{nr}$, we may think of $C$ as the row space of a $kr \times rn$ matrix $G$ over $\Fq$. The code $C$ can also be seen as a linear code of length $rn$ and rank $kr$ over $\Fq$, and thus as an almost affine code over the alphabet $\Fq$. We denote by $\rho_1$ and $\rho_r$ the rank functions of the almost affine codes $C$ over $F$ and $\Fq$ respectively. For $1 \leqslant x \leqslant n$, we also denote by $x_r$ the set \[x_r=\{(x-1)r+1,\cdots,(x-1)r+r\}\] and if $X \subset \{1,\cdots,n\}$, \[X_r=\bigcup_{x \in X} x_r.\] The rank functions $\rho_1$ and $\rho_r$ are given by \[\rho_1(Y) = \rk_ {\Fq} G_Y\] for $Y \subset \{1,\cdots,rn\}$. Also, for $X \subset \{1,\cdots,n\}$, \[\rho_r(X) = \frac{1}{r} \rk_{\Fq} G_{X_r} = \frac{1}{r}\rho_1(X_r).\]

The goal of this section is to show that a multilinear code $C$ in a natural way has a dual multilinear code. Interpreted as a linear code over $\Fq$, $C$ has a dual linear code $C^\perp$, namely the orthogonal complement of $C$ in $\Fq^{nr}$. Let $H$ be a generator matrix of $C^\perp$. This is a $(rn-kr) \times rn$ matrix over $\Fq$. Then, for $Y \subset \{1,\cdots,rn\}$, \[\rk_{\Fq} H_Y = |Y| +\rk_{\Fq}G_{\{1,\cdots,rn\} \- Y} -kr.\] In particular, for every $X \subset \{1,\cdots,n\}$, \begin{eqnarray*}\rk_{\Fq} H_{X_r} &= &|X_r| +\rk_{\Fq}G_{\{1,\cdots,rn\} \- X_r} -kr\\& =& r|X| + \rk_{\Fq} G_{\left(\{1,\cdots,n\}\-X\right)_r} -kr\\&=&r|X| + r\rho_m( C_{\{1,\cdots,n\}\-X}) -kr\end{eqnarray*} is divisible by $r$, and makes therefore $C^\perp$ a multilinear code.

\begin{remark}
As almost affine codes over the alphabet $\Fq^r$, the codes $C$ and $C^\perp$ have dual matroids. As a consequence, Wei duality holds for $C$ and $C^\perp$
\end{remark}

\section{Generalized Kung's bound}

In~\cite[Lemma 4.24]{K}, Kung gives a bound for the minimum number of codewords of a linear code that are sufficient to cover the whole space. This bound is related to the Singleton defect of the dual linear code. In~\cite{JSV}, this was generalized to find a bound for the number of codewords that are necessary to cover a subspace of the whole space. Both results rely heavily on linear algebra. In this section, we prove a similar result for almost affine codes.

We begin by defining the generalized  critical exponents. 
\begin{definition}\label{critical} Let $C$ be a non-degenerate almost affine code of length $n$. Let $\bm{\tilde{c}} \in C$ and $1 \leqslant i \leqslant n$. Then the $i$-th critical exponent with respect to $\bm{\tilde{c}}$ is \[\gamma_i(\bm{\tilde{c}}) = \min \{j,\ \exists \bm{c_1},\cdots,\bm{c_j} \in C,\ |\bigcup_{l=1}^j Supp(\bm{c_l},\bm{\tilde{c}})| \geqslant i\}.\]
\end{definition}

\begin{remark} If the dimension of $C$ is $k$, then it is obvious that \[\gamma_i(\bm{\tilde{c}}) = 1\ \forall 1\leqslant i\leqslant k\] since there exists at least a word of support $k$. Take namely a basis $B$ of $M_C$, then $C_B = F^{|B|}$ and we can find a word whose $\bm{\tilde{c}}$-support contains $B$.
\end{remark}

In~\cite{SA}, one can find the following result:

\begin{proposition}\label{prop6} The number of codewords in C with given $\bm{\tilde{c}}$-support $X$ is equal to \[\sum_{Y \subseteq X} (-1)^{|X \- Y|} q^{k-r(E \- Y)}\]
\end{proposition}
\begin{proof} This is~\cite[Proposition 6]{SA}.
\end{proof}

\begin{corollary} The generalized critical exponents are independent of the chosen word $\bm{\tilde{c}}$.
\end{corollary}

\begin{proof} Let $\bm{\tilde{d}} \in C$ be another word. Let $j=\gamma_i(\bm{\tilde{c}})$ and $\bm{c_1},\cdots\bm{c_j}$ be such that \[\left| \bigcup_{1\leqslant l \leqslant j} Supp(\bm{c_l},\bm{\tilde{c}})\right| \geqslant i.\] Let $X_i=Supp(\bm{c_i},\bm{\tilde{c}}).$ By definition, there exists at least one word, namely $\bm{c_i}$ whose $\bm{\tilde{c}}$-support is $X_i$. 
So, by the previous proposition, \[\begin{split}
\lefteqn{\left|\{\bm{w} \in C,\ Supp(\bm{w},\bm{\tilde{d}})=X_i\} \right|} \\& =\sum_{Y \subseteq X_i} (-1)^{|X_i \- Y|} q^{k-r(E \- Y)} \\ &= \left|\{\bm{w} \in C,\ Supp(\bm{w},\bm{\tilde{c}})=X_i\} \right| \\& \geqslant 1\end{split}\]%\[1\leqslant \left|\{\bm{w} \in C,\ Supp(\bm{w},\bm{\tilde{c}})=X_i\} \right| = \sum_{Y \subseteq X_i} (-1)^{|X_i \- Y|} q^{k-r(E \- Y)} = \left|\{\bm{w} \in C,\ Supp(\bm{w},\bm{\tilde{d}})=X_i\} \right|.\] 
Thus there exists a word $\bm{d_i} \in C$ such that $Supp(\bm{d_i},\bm{\tilde{d}}) = X_i$. Then \[\left| \bigcup_{1\leqslant l \leqslant j} Supp(\bm{d_l},\bm{\tilde{d}})\right|=\left| \bigcup_{1\leqslant l \leqslant j} Supp(\bm{c_l},\bm{\tilde{c}})\right| \geqslant i,\] and this shows that \[\gamma_i(\bm{\tilde{d}}) \leqslant  \gamma_i(\bm{\tilde{c}})\] and equality comes by symmetry.
\end{proof}

In the sequel, we will therefore omit the reference to a particular word in the critical exponents.

Before stating and proving the main result of this section, we need a lemma on matroid theory.
\begin{lemma}\label{lemmematroide}Let $M$ be a matroid on the ground set $E$. Let $B$ a basis and $x \in E\-B$. Then for every $y \in B$, we have: $B' = B \-\{y\} \cup \{x\}$ is a basis of $M$ if and only if $y \in \sigma(B,x)\-\{x\}$
\end{lemma}

\begin{proof}
Assume that $B'$ is not a basis. Then $\rho(B') \neq |B'|$, and by a repeated use of axiom $(R2)$, $\rho(B') < |B'|$. By the same axiom again, since $\rho(B)=|B|$, we get successively $\rho(B \-\{y\}) = |B|-1$ and $\rho(B')=\rho(B\-\{y\})=|B|-1=|B'|-1$. This shows that $B'$ contains a circuit, say $\tau$. Of course, this circuit contains $x$, otherwise it is contained in $B$, and a repeated use of axiom $(R2)$ again would show that any subset of $B$ has rank equal to its cardinality. Thus, $\tau$ is a circuit contained in $B \cup \{x\}$, and by Lemma~\ref{fundamental}, $\tau=\sigma(B,x)$. Since $y \not \in \tau$, one way is shown.\\
Assume now that $y \not \in \sigma(B,x)$. Then \[\sigma(B,x) \subset B' = B \-\{y\} \cup \{x\}.\] Since $\rho(\sigma(B,x)) = |\sigma(B,x)|-1$, by a repeated use of axiom $(R2)$ again, \begin{eqnarray*}\rho(B') &=& \rho(\sigma(B,x) \cup (B'\- \sigma(B,x)) \\&\leqslant& \rho(\sigma(B,x)) + |B' \- \sigma(B,x)|\\& \leqslant& |\sigma(B,x)|-1 + |B' \- \sigma(B,x)| \\&\leqslant& |B'|-1\end{eqnarray*} and $B'$ is not a basis.
\end{proof}

\begin{theorem}\label{Kungaa}Let $C$ be a non-degenerate almost affine code of dimension $k$ and length $n$ on the alphabet $F$. Let $k+1 \leqslant i\leqslant n$. Then we have \[\gamma_i \leqslant s^*_{n+1-i}+2\] where $s^*_j$ denotes the $j$-th generalized Singleton defect of $M_C$, \[s^*_j=k+j-d^*_j.\] 
\end{theorem}

\begin{remark} We recall that the generalized Hamming weights $d_i$ of the almost affine code $C$ are defined as the generalized Hamming weights for the dual $M^*_C$ of $M_C$. From Wei duality, we get the dual generalized Hamming weights $d^*_i$ of the code $C$ - and these do not in  general correspond to the Hamming weights for an almost affine code, since we have not been able to define dual codes of almost affine codes in general. If we think of matroids, these latter weights correspond to generalized Hamming weights for the matroid $M_C$, that is \[d^*_j = \min \{|X|,\ n(X) = j\}.\] In the special case that $C$ is a linear code over $\mathbb{F}_q$, then these $d^*_i$ are the usual Hamming weights for the orthogonal complement $C^{\perp}$, and we obtain (a new proof of)~\cite[Theorem 9]{JSV}.
\end{remark}

\begin{proof}[Proof of Theorem~\ref{Kungaa}] Let $q=|F|$. Without loss of generality, we may assume that the alphabet is $F=\{0,\cdots,q-1\}$, that $\bm{\tilde{c}} = (0,\cdots,0)$, and that $B=\{1,\cdots,k\}$ is a basis of $M_C$. 
By Corollary~\ref{surbase}, there exists for each $1\leqslant j \leqslant k$ a unique word $\bm{w^{(j)}} \in C$ such that $\bm{w^{(j)}}_l=0$ for $l \in \{1,\cdots,k\}\-\{j\}$ and $\bm{w^{(j)}}_j=1$. Now, let $S \subset \{k+1,\cdots,n\}$ be of cardinality $n+1-i$, and set \[T_S=\{l \in \{1,\cdots,k\},\ \exists j \in S,\ \bm{w^{(l)}}_j \neq 0\}.\] We claim that \[|T_S| \geqslant d_{n+1-i}-(n+1-i).\] Indeed, let $j \in S$ and $l \in \sigma(B,y)\-\{y\}$. This latter is non-empty since the code is non-degenerate and thus the matroid $M_C$ has no loops. By Lemma~\ref{lemmematroide}, $B_l=B\-\{j\} \cup \{l\}$ is still a basis of $M_C$. By Proposition~\ref{prop2}, the almost affine subcode $C(B_l,\bm{\tilde{c}})$ is such that \[|C(B_l,\bm{\tilde{c}})|=q^{k-r(B_l)} = 1\] Since $\bm{\tilde{c}} \in C(B_l,\bm{\tilde{c}})$, this means that $\bm{w^{(l)}} \not \in C(B_l,\bm{\tilde{c}})$, and in particular $\bm{w^{(l)}}_j \neq 0$. This shows that  \[\bigcup_{j \in S} \left(\sigma(B,j)\-\{j\}\right) \subset T_S\] and therefore \[|T_S| \geqslant \left| \bigcup_{j \in S} \left(C(B,j)\-\{j\}\right)\right| = \left|\bigcup_{j \in S} \sigma(B,j)\right| -|S|.\] Now, the circuits $\sigma(B,j)$ are non-redundant, so from Proposition~\ref{supp}, we know that \[n\left(\bigcup_{j \in S} \sigma(B,j)\right) \geqslant |S| = n+1-i.\] This in turn implies that \[\left|\bigcup_{j \in S} \sigma(B,j)\right| \geqslant d^*_{n(\bigcup_{j \in S} \sigma(B,j))} \geqslant d^*_{n+1-i},\] the first inequality coming from the definition \[d^*_l=\min\{|X|,\ n^*(X)=l\}\] and the second inequality from the monotony property of generalized Hamming weights.\\
Now, if we take $t=k+n+2-i-d^*_{n+1-i}$ distinct words among $(\bm{w^{(1)}},\cdots,\bm{w^{(k)}})$, say $\bm{w^{(l_1)}},\cdots,\bm{w^{(l_t)}}$, then we claim that \[\left|\bigcup_{1\leqslant s \leqslant t} Supp(\bm{w^{(l_s)}},\bm{\tilde{c}}) \cap \{k+1,\cdots,n\} \right|\geqslant i-k.\] If not, then there would exist at least $n+1-i$ distinct indices $j$ in $\{k+1,\cdots,n\}$ such that \[\forall 1 \leqslant s \leqslant t,\ \bm{w^{(l_s)}}_j = 0.\] Take $S$ to be $n+1-i$ such indices. Then for this particular $S$, we would have \[|T_S| \leqslant k-t<d^*_{n+1-i}-(n+1-i)\] which is absurd.\\
These $t$ words, together with the word $\bm{w_0} \in C$ such that $\bm{(w_0)}_B=(1,\cdots,1)$ gives a $t+1$-tuple whose support has cardinality at least $i$, and this concludes the proof.
\end{proof}

\begin{remark}These bounds are the best that can be found. Linear codes are namely almost affine codes, and in~\cite{JSV}, it is mentioned that for simplex codes, the bounds are reached.
\end{remark}

\begin{example}Let  $C'$ be the code of Example~\ref{running}. Let $\bm{\tilde{c}}= 321$. Then $\gamma_3(\bm{\tilde{c}})=1$ since $Supp(213,\bm{\tilde{c}})=\{1,2,3\}$. We have seen that $d_1(C')=2$ and $d_3(C')=3$, so that by Wei duality, $d^*_1(C') = 3$. The bound of theorem~\ref{Kungaa} says that \[1=\gamma_3(\bm{\tilde{c}}) \leqslant s^*_1(C')+2 =2.\]  
\end{example}

\section{Profiles of almost affine codes and trellis decoding}

In~\cite{Mu}, Muder describes trellis decoding for block codes. In~\cite{F}, Forney defines various dimension/length profiles for linear codes. These profiles give a lower bound for the complexity of the minimal trellis associated to the code, and thus an indication on how well decoding using the Viterbi algorithm will work. 

In this section, we observe how the Viterbi algorithm immediately works for almost affine codes, and we show how the dimension/length profile concept can be generalized to these codes as well and how they are related to the generalized Hamming weights. For self-containment and clarity, we include the trellis decoding algorithm.

\subsection{Dimension/length profiles and generalized Hamming weights}

\begin{definition}
The dimension/length profile of an almost affine code $C$ of dimension $k$ and length $n$ is the sequence $k_i(C)$ for $1 \leqslant i \leqslant n$ where \[k_i(C) = \max \left \{ \begin{array}{c}\dim D,\ D \subset C \textrm{ is an almost affine}\\\textrm{code with }|Supp(D)| \leqslant i\end{array}\right\}.\]
\end{definition}

In the definition above, we can actually restrict to subcodes of the type $C(X,\bm{\tilde{c}})$:

\begin{proposition}Let $\bm{\tilde{c}} \in C$.
We have \[k_i(C) = \max\left\{ \log_q\left|C(X,\bm{\tilde{c}})\right|, |X|=n-i\right\}.\]
\end{proposition}

\begin{proof}
It is clear that $Supp(C(X,\bm{\tilde{c}})) \subset E\-X$, so that $\left|Supp(C(X,\bm{\tilde{c}}))\right| \leqslant i.$ This proves that \[k_i(C) \geqslant \max\left\{ \log_q\left|C(X,\bm{\tilde{c}})\right|, |X|=n-i\right\}.\]
On the other hand, let $D \subset C$ be an almost affine subcode such that $|Supp(D)| \leqslant i$ and $\dim D = k_i(C)$. Let $X = Supp(D)$ and $X \subset Y \subset E$ be such that $|Y|=i$. Consider $D' = C(E\-Y,\bm{\tilde{c}})$ for any $\bm{\tilde{c}} \in D$. Obviously $|Supp(D')| \leqslant i$ and $\dim D' \geqslant \dim D = k_i(C)$ proving the proposition.
\end{proof}

\begin{corollary}\label{cor13}
We have \begin{eqnarray*}k_i(C) &=& \max\left\{ k-r(X), |X|=n-i\right\}\\& =& k-\min \left\{r(X), |X|=n-i\right\}.\end{eqnarray*}
\end{corollary}

The dimension/length profile is related to the generalized Hamming weights of the code in the following way:

\begin{proposition}
We have \[d_j(C) = \min\{i,\ k_i(C) \geqslant j\}\] and \[k_i(C) = \max \{j,\ d_j(C) \leqslant i\}.\]
\end{proposition}
\begin{proof}
We have \[\begin{split}\lefteqn{\min\{i ,\  k_i(C) \geqslant j \}} \\&=\min\{i ,\ \max\{\log_q|C(X,\bm{\tilde{c}})| ,\ |X| = n-i\} \geqslant j \}\\&=n- \max\{i ,\ \max\{\log_q|C(X,\bm{\tilde{c}})| ,\ |X| = i\} \geqslant j \} \\&=n- \max\{ |X| ,\ \log_q|C(X,\bm{\tilde{c}})| \geqslant  i\} \\&=n- \max\{ |X| ,\ \log_q|C(X,\bm{\tilde{c}})| = i\} \\&=d_j(C),\end{split}\]the penultimate equality coming from the fact that $\log_q|C(X,\bm{\tilde{c}})|$ decreases by at most $1$ if $X$ is augmented with $1$ element.

Moreover we have:

\[\begin{split}\lefteqn{\max \{j,\ d_j(C) \leqslant i\}}\\&=\max\{j ,\ ( n - \max\{ |X| ,\  \log_q|C(X,\bm{\tilde{c}})|=j\}) \leqslant i \}\\ &=\max\{j ,\ ( \max\{ |X| ,\  \log_q|C(X,\bm{\tilde{c}})| =j\}) \geqslant n-i \}\\&= \max\{ \log_q|C(X,\bm{\tilde{c}})| ,\ |X| \geqslant   n-i \}
\\&= \max\{ \log_q|C(X,\bm{\tilde{c}})| ,\ |X| =  n-i \}\\&=k_i(C).
\end{split}\]
\end{proof}

\subsection{Trellis decoding for almost affine codes}

\begin{definition}
A proper trellis is a labelled directed graph such that the vertices can be partitioned into subsets $V_0, \cdots, V_n$ such that the only possible directed edges are between an element in $V_i$ and an element in $V_{i+1}$. Moreover, $|V_0|=|V_n| =1$, and every vertex in $V_i$ for $1\leqslant i \leqslant n-1$ is connected to at least one vertex in $V_{i-1}$ and one vertex in $V_{i+1}$. It is proper when no two edges from the same vertex have the same label. We say that it represents $C$ if $C$ is equal to the set of concatenations of the labels of the edges of paths from $V_0$ to $V_n$. It is minimal if it has fewer vertices at every stage than any other proper trellis representing $C$.
\end{definition}

Let $C$ be an almost affine code of dimension $k$ and length $n$ on an alphabet of $F$ cardinality $q$. We define a labelled directed graph $G=(V,T)$ in the following way. For $0\leqslant i \leqslant n$, let $C_i=C_{\{1,\cdots,i\}}$. In particular, $C_0=\{\emptyset\}$ and $C_n=C$. We define an equivalence relation $\sim$ on $C_i$ by: for $\bm{v},\bm{w} \in C_i$, let $\bm{v'}, \bm{w'} \in C$ be such that $\bm{v'}_{\{1,\cdots,i\}} = \bm{v}$ and $\bm{w'}_{\{1,\cdots,i\}} = \bm{w}$, 

\[\bm{w} \sim \bm{v} \Leftrightarrow \begin{array}{c}C(\{1,\cdots,i\},\bm{v'})_{\{i+1,\cdots,n\}}\\\shortparallel\\{C(\{1,\cdots,i\},\bm{w'})_{\{i+1,\cdots,n\}}}\end{array}.\] 

It is independent of the choice of $\bm{v'}$ and $\bm{w'}$. In other words, $\bm{v}$ and $\bm{w}$ are equivalent if and only if every ending of a word in $C$ starting with $\bm{v}$ is an ending of a word in $C$ starting with $\bm{w}$, and conversely. We denote by $[\bm{v}]_i$ the equivalence class of $\bm{v}$. Let $V_i= C_i/_\sim,$ for $0\leqslant i \leqslant n.$ In particular, $V_0=\{[\emptyset]_0\}$ and $V_n=\{[\bm{w}]_n\}$ for any $\bm{w} \in C$. The set of vertices of $G$ is then defined by $V= \bigcup_{i=0}^n V_i$. The set of labelled edges is \[T= \left\{ ([\bm{v}]_i,[\bm{w}]_{i+1},\alpha), \begin{array}{c}\exists\bm{v'} \in [\bm{v}]_i,\ \exists \bm{w'} \in [\bm{w}]_{i+1},\\ \bm{w'}=\bm{v'}| \alpha\end{array}\right\},\] where $\bm{v'}|\alpha$  is the concatenation of $\bm{v'}$ and $\alpha$, and $\alpha$ is the label on the edge. One can show that this graph is a minimal proper trellis representing $C$.

\begin{example}\label{exmintrel}
Let $C$ be the code from Example~\ref{running}. Then $V_0=\{[\emptyset]_0\}.$ $V_1=C_1=\{[0]_1,[1]_1,[2]_1,[3]_1\}$. Namely, the ending of the words beginning with $0$ ($00,11,22,33$) are different than the endings of the word starting with $1$ ($01,12,23,30$) and so on. It is different for $V_2$. Namely, all the words beginning with $00,31,22,13$ have the same ending, namely $0$, so they are in the same equivalence class. We get that $V_2=\{[00]_2,[11]_2,[12]_2,[23]_2\}$. Finally, $V_3=\{[000]_3\}$. For the edges, there is for example one edge going from $[\emptyset]_0$ to $[1]_1$, with label $1$. There is also one edge going from $[1]_1$ to $[00]_2$ with label $3$. Namely, $1 \in [1]_1$, $13 \in [00]_2$ and $13=1|3$.
The minimum trellis representing $C$ is the following, where the plain, dotted, dashed and wave arrows are labelled with $0$, $1$, $2$ and $3$ respectively:

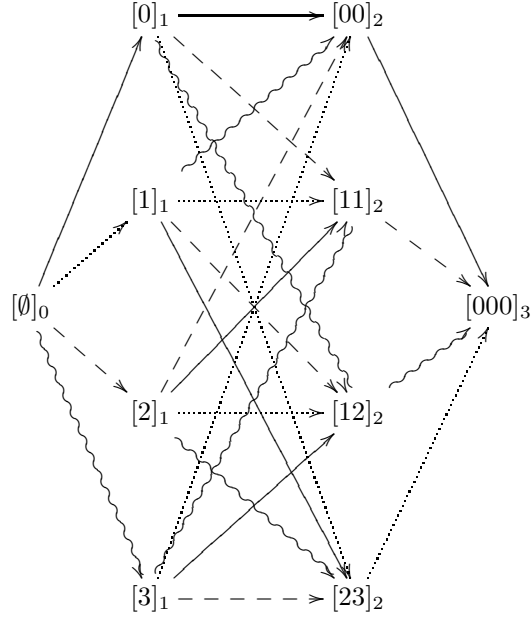
\[\xymatrix{ & [0]_1 \ar@{.>}[rrdddddd] \ar@{~>}[rrdddd] \ar@{-->}[rrdd] \ar[rr] && [00]_2\ar[rddd] & \\ \\ & [1]_1\ar@{~>}[rruu] \ar@{.>}[rr] \ar@{-->}[rrdd] \ar[rrdddd]&& [11]_2 \ar@{-->}[rd]& \\ [\emptyset]_0 \ar[ruuu] \ar@{.>}[ru]\ar@{-->}[rd]\ar@{~>}[rddd]&&&& [000]_3 \\& [2]_1\ar@{-->}[rruuuu] \ar[rruu] \ar@{.>}[rr] \ar@{~>}[rrdd] &&  [12]_2 \ar@{~>}[ru] &\\ \\ & [3]_1\ar@{.>}[rruuuuuu] \ar@{~>}[rruuuu] \ar[rruu] \ar@{-->}[rr] && [23]_2\ar@{.>}[ruuu] &}\]
%\[\xymatrix{ & [0]_1 \ar@{.>}[rrrrdddddd] \ar@{~>}[rrrrdddd] \ar@{-->}[rrrrdd] \ar[rrrr] &&&& [00]_2\ar[rddd] & \\ \\ & [1]_1\ar@{~>}[rrrruu] \ar@{.>}[rrrr] \ar@{-->}[rrrrdd] \ar[rrrrdddd]&&&& [11]_2 \ar@{-->}[rd]& \\ [\emptyset]_0 \ar[ruuu] \ar@{.>}[ru]\ar@{-->}[rd]\ar@{~>}[rddd]&&&&&& [000]_3 \\& [2]_1\ar@{-->}[rrrruuuu] \ar[rrrruu] \ar@{.>}[rrrr] \ar@{~>}[rrrrdd] &&& & [12]_2 \ar@{~>}[ru] &\\ \\ & [3]_1\ar@{.>}[rrrruuuuuu] \ar@{~>}[rrrruuuu] \ar[rrrruu] \ar@{-->}[rrrr] &&&& [23]_2\ar@{.>}[ruuu] &}\]
\end{example}
Any trellis representing $C$, and thus this minimal trellis, can be used for decoding, using the Viterbi algorithm (\cite{V}). Given a word $\bm{c} \in F^n$, the algorithm finds the words in $C$ such that their Hamming distance to $\bm{c}$ is minimal. The algorithm runs as follows:

\begin{algorithmic}
\State $W \gets \{\bm{\emptyset}\}$
\For{$1 \leqslant i \leqslant n$}
\State $W' \gets \emptyset$
\ForAll{$[\bm{v}]_i \in V_i$}
\State $H \gets \{\bm{w}|\alpha, \bm{w} \in W, (End(\bm{w}),[\bm{v}]_i,\alpha) \in T\}$
\State $H \gets \{\bm{w} \in H,\ d(\bm{w},\bm{c}_{\{1,\cdots,i\}}) \ minimal\}$
\State $W' \gets W' \cup H$
\EndFor
\State $W \gets W'$
\EndFor
\State \Return $W$
\end{algorithmic}

Here, if $\bm{w} \in C_i$, $End(\bm{w})$ is the unique edge corresponding to the path from $[\bm{\emptyset}]_0$ and label $\bm{w}$. In the previous example, $End(20) = [11]_2$.

We will not do an analysis of the Viterbi algorithm. The idea of why it works is that whenever one comes to a node $[\bm{v}]_i \in V_i$, one can keep the words ending there that have minimal Hamming distance with $\bm{c}_{\{1,\cdots,i\}}$. Namely, all the other words ending there will have a strictly larger Hamming distance in further stages, since the possible endings of all these words are all the same (by definition of the equivalence relation).

\begin{example}\label{exviterbi}
We continue with Example~\ref{exmintrel}. Suppose that we receive the word $320$. In the first loop, we keep all the words of length $1$ (each vertex has just one incoming edge, and it must be kept). In the second loop, we look first at the vertex $[00]_2$. It has $4$ incoming edges, that give the following words: $00,13,22,31$, with Hamming distance $2,2,1,1$ to $32$ respectively. So we just keep the two last ones, namely $22$ and $31$. For the vertex $[11]_2$ we keep the words $02,33$, for the vertex $[12]_2$ we keep the words $12,30$, all of them having Hamming distance  $1$ to $32$. For the vertex $[23]_2$, we keep only $32$, with Hamming distance $0$ to $32$. For the third loop, there are $4$ incoming edges to $[000]_3$, and this leads to the following words to look at: $220,310,022,332,123,303,321$. We keep those with minimal Hamming distance to $322$, namely: $022, 332, 321$.
\end{example}

The complexity of the algorithm is related to the number of vertices at each stage, that is $|V_i|$. Here, we give a minimal bound for this number.

\begin{proposition}
For every $1 \leqslant i\leqslant n$, \[\log_q|V_i| \geqslant k-k_i(C)-k_{n-i}(C).\]
\end{proposition}

\begin{proof}
Let $\bm{v} \in C_i$ and $\bm{w} \in C$ such that $\bm{v}=\bm{w}_{\{1,\cdots,i\}}$. Let $\bm{t}=\bm{w}_{\{i+1\cdots,n\}}$. Let $\bm{c} \in C_i$. Then if $\bm{c} \in [\bm{v}]_i$, it implies that $\bm{c}|\bm{t} \in C$. In particular, \[\bm{c}|\bm{t} \in C(\{i+1,\cdots,n\},\bm{w})\] In turn, this implies that \[|[\bm{v}]_i| \leqslant |C(\{i+1,\cdots,n\},\bm{w})| = q^{k-r(\{i+1,\cdots,n\})}.\] Now, $C_i$ is a disjoint union of these equivalence classes, and has cardinality $q^{r(\{1,\cdots,i\})}$ so that we get that \[|V_i| \geqslant \frac{q^{r(\{1,\cdots,i\})}}{q^{k-r(\{i+1,\cdots,n\})}}.\] Thus, by Corollary~\ref{cor13} \begin{eqnarray*}\log_q|V_i| &\geqslant& r(\{1,\cdots,i\}) + r(\{i+1,\cdots,n\})-k \\&\geqslant& \min\{r(X),\ |X|=i\} \\&&+ \min\{r(X),\ |X|=n-i\} - k \\&\geqslant& k-k_i(C)-k_{n-i}(C).\end{eqnarray*}
\end{proof}

\begin{remark}
It would have been beneficial to have upper bounds, and not only lower bounds, for the complexity of the trellis decoding algorithm. But as far as we know, no such non-trivial bounds are known, even for linear codes.
\end{remark}

\section{Wire-tap channel of type II}\label{VI}

In~\cite{OW}, Ozarow and Wyner introduce the wire-tap channel of type II. A sender wants to send $k$ elements of information. In order to do so, the information is encoded into $n$ elements, and sent to the receiver. An intruder is allowed to listen to any $s$ elements of the sent message. The channel is noiseless, so the receiver can decode the message correctly. The authors look at how much information the intruder is able to get. In their paper, they present an encoder/decoder system using linear codes. In~\cite{W}, Wei relates the equivocation (that is, a measure on the minimum of uncertainty for an intruder about the source) of the system to the generalized Hamming weights for the code (and its dual code).

In this section, we extend their results to almost affine codes. We show that we can use almost affine codes to design an encoder/decoder system, and we relate the equivocation of the system to the generalized Hamming weights for the dual of the matroid associated to the almost affine code.

So let $C$ be an almost affine code on the alphabet $F$ with $|F|=q$, of dimension $k$ and length $n$. Without loss of generality, we may assume that the set $B=\{1,\cdots,k\}$ is a basis of the associated matroid $M_C$. Let $\varphi: F^{n-k} \times F^{n-k} \rightarrow F^{n-k}$ be a mapping such that for all $\bm{f} \in F^{n-k}$, $\varphi(\bm{f},.)$ is a bijection and such that \[\begin{split}\forall X \subset \{1,\cdots,n-k\}, \forall \bm{m},\bm{f},\bm{g} \in F^{n-k},\\ \bm{g}|_X = \bm{h}|_X \Leftrightarrow \varphi(\bm{g},\bm{m})|_X = \varphi(\bm{h},\bm{m})|_X.\end{split}\] \begin{remark} All these conditions are true if $\varphi_0: F \times F \longrightarrow F$ is a mapping such that $\varphi_0(x,.): F \longrightarrow F$ is a bijection for every $x \in F$, and $\varphi : F^{n-k} \times F^{n-k} \longrightarrow F^{n-k}$ is defined by \[\begin{split}\lefteqn{\varphi((a_1,\cdots,a_{n-k}),(b_1,\cdots,b_{n-k}))} \\&= (\varphi_0(a_1,b_1),\cdots,\varphi_0(a_{n-k},b_{n-k})).\end{split}\] \end{remark}Extend $\varphi$ to $\tilde{\varphi}: F^n \times F^{n-k} \rightarrow F^n$ in the following way: for every $\bm{f} \in F^n$ and $\bm{g} \in F^{n-k}$, \[\tilde{\varphi}(\bm{f},\bm{g})_i = \left\{\begin{array}{ll} \bm{f}_i & \textrm{ if } 1 \leqslant i \leqslant k,\\ \varphi(\bm{f}|_{E \- B},\bm{g})_{i-k} & \textrm{ otherwise} \end{array}\right.\] For every $\bm{m}\in F^{n-k}$, define \[C_{\varphi,\bm{m}} = \{\tilde{\varphi}(\bm{w},\bm{m}),\ \bm{w} \in C\}.\]
When $\varphi$ is obvious from the context, we will omit it and write $C_{\bm{m}}$ for $C_{\varphi,\bm{m}}$.

\begin{lemma} The sets $\{C_{\bm{m}},\ \bm{m} \in F^{n-k}\}$ form a partition of $F^n$.
\end{lemma}

\begin{proof}It is obvious that there is a bijection between $C_{\bm{m}}$ and $C$, since $\tilde{\varphi}(.,\bm{m})$ is a bijection when restricted to $C$, since it leaves the coordinates on a basis unchanged.  Now, suppose that $\bm{c}=(c_1,\cdots,c_n) \in C_{\bm{m}} \cap C_{\bm{m'}}$. In particular, we have that  
\[(c_1,\cdots,c_k,c_{k+1},\cdots,c_n) = \tilde{\varphi}(\bm{w},\bm{m}) = \tilde{\varphi}(\bm{w'},\bm{m'})\]
for some words $\bm{w},\bm{w'} \in C$. Then $\bm{w}|B = \bm{w'}|B$, and by Proposition~\ref{prop2}, $\bm{w}=\bm{w'}$. On the other hand, we have \[\begin{split}\lefteqn{\varphi(\bm{w},\bm{m}) = \tilde{\varphi}(\bm{w},\bm{m})|_{E \- B}}\\& = \tilde{\varphi}(\bm{w'},\bm{m'})|_{E \- B}=\tilde{\varphi}(\bm{w},\bm{m'})|_{E \- B} = \varphi(\bm{w},\bm{m'})\end{split}\]which implies that $\bm{m}=\bm{m'}$ since $\varphi(\bm{w},.)$ is a bijection. We conclude by a cardinality argument.
\end{proof}
\begin{lemma} The sets $C_{\bm{m}} \subset F^n$ are almost affine codes with associated matroid $M_C$.
\end{lemma}

\begin{proof} %We have namely, for $X \subset \{1,\cdots,n\}$, \[(C_{\varphi,\bm{m}})_X = (C_X)_{\overline{\varphi},\bm{\overline{m}}}\] where $\bm{\overline{m}}=\bm{m}|_{X \cap (E \- B)}$ and $\overline{\varphi} = \varphi|_{(X \cap (E \- B))\times (X \cap (E \- B))}$. Then \[\left|(C_{\varphi,\bm{m}})_X\right| =\left| (C_X)_{\overline{\varphi},\bm{\overline{m}}}\right| = \left|C_X\right|.\]
Let $X \subset \{1,\cdots,n\}$ and $Y=X \cap B$, $Z=X \- Y$. We will construct a bijection \[\theta: C_X \longrightarrow \left(C_{\bm{m}}\right)_X\] in the following way: let $\bm{v} \in C_X$ and $\bm{w} \in C$ such that $\bm{w}|_X=\bm{v}$. Then let $\theta(\bm{v}) = \tilde{\varphi}(\bm{w},\bm{m})|_X.$ This is well defined since if $\bm{w},\bm{w'} \in C$ are such that $\bm{w}|_X = \bm{w'}|_X$,  then  $\bm{w}|_Z = \bm{w'}|_Z$. This in turn implies that $\varphi(\bm{w}|_{E \-B},\bm{m})|_Z = \varphi(\bm{w'}|_{E \-B},\bm{m})|_Z$, and thus, combined with the fact that $\bm{w}|_Y = \bm{w'}|_Y$, $\tilde{\varphi}(\bm{w},\bm{m})|_X=\tilde{\varphi}(\bm{w'},\bm{m})|_X$.

This is injective because if $\bm{v_1},\bm{v_2} \in C_X$ are such that $\bm{v_1} \neq \bm{v_2}$, let $\bm{w_1},\bm{w_2} \in C$ be such that $\bm{w_1}|_X=\bm{v_1}$ and  $\bm{w_2}|_X=\bm{v_2}$. Then at least one of the two cases is true: \begin{itemize}
\item $\bm{w_1}|_Y \neq \bm{w_2}|_Y$ and then trivially $\tilde{\varphi}(\bm{w_1},\bm{m})|_X\neq\tilde{\varphi}(\bm{w_2},\bm{m})|_X$
\item $\bm{w_1}|_Z \neq \bm{w_2}|_Z$. Then $\varphi(\bm{w_1}|_{E \- B},\bm{m})|_Z \neq \varphi(\bm{w_2}|_{E \- B},\bm{m})|_Z $, and in turn $\tilde{\varphi}(\bm{w_1},\bm{m})|_X\neq\tilde{\varphi}(\bm{w_2},\bm{m})|_X$.
\end{itemize} 
Surjectivity is obvious by construction.

Then, \[\left|\left(C_{\bm{m}}\right)_X\right| = |C_X|\] which proves the lemma.
\end{proof}

Our scheme is then the following: the encoder wants to send the message $\bm{m} \in F^{n-k}$, and chooses randomly and uniformly any element $\bm{c} \in C_{\bm{m}}$, and sends it. The decoder gets $\bm{c}\in F^n$, finds the unique codeword $\bm{w} \in C$ such that $\bm{w}|_B=\bm{c}|_B$. Then $\bm{m} \in F^{n-k}$ is the unique element such that $\varphi(\bm{w}|_{E \- B},\bm{m})=\bm{c}|_{E \- B}$.

If the message $\bm{t} \in F^n$ is sent over the channel, and an intruder is able to listen to a subset $X \subset \{1,\cdots,n\}$ of the digits of $\bm{t}$, we will now see how much the intruder knows about $\bm{m}$, namely which $\bm{m}$ the sender could possibly have tried to send, and with which probability.

\begin{example}\label{debut}Let $C'$ be the code of Example~\ref{running}. Here the alphabet is $\{0,1,2,3\}$, and we take $\varphi({a},{b})={a}+{b} \ (mod\ 4)$. We want to send the message $\bm{m}=2$. We therefore construct  $C'_{\bm{2}}$: \begin{align*} 002 && 013 && 020 && 031\\103&&110&&121&&132 \\ 200&&211&&222&&233\\301&&312&&323&&330\end{align*} We choose at random any element there, say $121$ and send it to the receiver. The receiver sees that the only word in $C'$ starting with $12$ is $123$, so that the message that was sent is $\bm{m}$ such that $\bm{m}+3=1$, that is $\bm{m}=2$.\\
An intruder able to listen to $1$ digit, say the second, knows nothing about $\bm{m}$. Namely, there are exactly $4$ elements in $C'_{\bm{2}}$ such that the second digit is $2$, but the same is true also for $C'=C'_{\bm{0}}$, $C'_{\bm{1}}$ and $C'_{\bm{3}}$. The same is true if the intruder is able to listen to $2$ digits, say the first and third. There is exactly $1$ word in each of $C'_{\bm{0}}$, $C'_{\bm{1}}$, $C'_{\bm{2}}$  and $C'_{\bm{3}}$ looking like $(1\cdot 1)$, namely $101$, $131$, $121$ and $111$ respectively.
\end{example}

\begin{lemma} Let $\bm{t} \in F^n$ be any word, and $X \subset \{1, \cdots, n\}$. Then we have the following \begin{itemize}
 \item Let $\bm{m} \in F^{n-k}$. Then the set \[\Lambda_{\bm{t},X}(\bm{m})=\{\bm{w} \in C_{\bm{m}},\ \bm{w}_X=\bm{t}_X\}\] is either empty, or has cardinality $\left|F\right|^{k-r(X)}.$
\item \[\left| \{\bm{m} \in F^{n-k},\ \Lambda_{\bm{t},X}(\bm{m}) \neq \emptyset\} \right|= \left|F\right|^{n-k-n(X)}.\]
\end{itemize}
\end{lemma}

\begin{proof}
Let's assume that $\Lambda_{\bm{t},X}(\bm{m}) \neq \emptyset$, and let $\bm{s} \in \Lambda_{\bm{t},X}(\bm{m})$. In particular, $\bm{s} \in C_{\bm{m}}$, and we have \begin{eqnarray*} |\Lambda_{\bm{t},X}(\bm{m})| &=& \left| \{\bm{w} \in C_{\bm{m}},\ \bm{w}_X=\bm{t}_X\}\right| \\&=& \left| \{\bm{w} \in C_{\bm{m}},\ \bm{w}_X=\bm{s}_X\}\right| \\&=&\left|C_{\bm{m}}(X,\bm{s})\right| \\&=& \left|F\right|^{\rk(C_{\bm{m}}) - r_{C_{\bm{m}}}(X)} \\&=& \left|F\right|^{k-r(X)}.\end{eqnarray*} 
For the second point of the proof, since \[|\{\bm{w} \in F^n,\ \bm{w}|_X = \bm{t}|_X\}| = |F|^{n-|X|},\] and all $C_{\bm{m}}$ are disjoint, each such $\bm{w}$ must be in a different set $\Lambda_{\bm{t},X}(\bm{m})$. We conclude using the first point.

\end{proof}

In particular, if $|X|<d_1^* = \min \{|X|,\ n(X)=1\}$, then an intruder that is able to listen to the subset $X$ of digits of $\bm{t}$ gets no information whatsoever on the message $\bm{m}$. Namely, for every $\bm{m'} \in F^{n-k}$, there are exactly $\left|F\right|^{k-|X|}$ words in $C_{\bm{m'}}$ whose restriction to $X$ is $\bm{t}_X$.

A way of measuring how much an intruder gains information is the conditional entropy of the system, namely \[\begin{split}\lefteqn{H(F^{n-k}|T_X)}\\& = - \sum_{\bm{t}_X \in T_X}p(\bm{t}_X) \sum_{\bm{m} \in F^{n-k}} p(\bm{m}|\bm{t}_X)\log_{|F|}p(\bm{m}|\bm{t}_X)\end{split},\] where $T_X$ is the set of possible observations made by the eavesdropper at places $X \subset \{1,\cdots,n\}$.  Now, we assume that all messages $\bm{m}$ have the same probability to be chosen, and then that the sent message $\bm{w} \in C_m$ the same probability to be chosen, so that $p(\bm{t}_X) = \frac{1}{|F|^X}.$ From the previous lemma, we have that  
\[p(\bm{m}|\bm{t}_X) = \begin{cases} 0  &\textrm{if } \Lambda_{\bm{t},X}(\bm{m}) = \emptyset \\ \frac{1}{|F|^{n-k -n(X)}} & \textrm{otherwise}   \end{cases}.\] This gives that \[H(F^{n-k}|T_X) = n-k- n(X).\] 

The system designer is interested in maximizing the equivocation \[E_\mu = \min_{|X|=\mu} H(F^{n-k}|T_X)\] for all possible $\mu \in \{0,\cdots,n\}.$ This way, the designer is assured that no matter which $\mu$ digits an intruder is able to listen to, the uncertainty about the message $\bm{m}$ is at least $E_\mu$. The maximum of information gained by an intruder with $\mu$ taps is therefore \[\Delta_\mu=n-k-E_\mu = \max_{|X|=\mu}\{n(X)\}.\]

By the definition of the generalized Hamming weights for the dual of the matroid $M_C$ associated to the code $C$, \[d_i^* = \min\{|X|, n(X) = i\},\] we get that \[\max_{|X| = \mu} \{n(X)\} = j \Leftrightarrow d_{j}^* \leqslant \mu < d_{j+1}^*,\]with the convention that $d^*_0=0$ and $d^*_{n-k+1}=n+1$.  We get then the following characterization of the equivocation of the system: 

\begin{theorem} The quantity $\Delta_\mu$ of the system described above is entirely determined by the dual generalized Hamming weights for the almost affine code $C$, namely \[d_{\Delta_\mu}^* \leqslant \mu < d_{\Delta_\mu+1}^*\] with the same convention as above.
\end{theorem}

\begin{example}We continue with Example~\ref{debut}. Since the matroid associated to $C'$ is $U_{3,2}$, the nullity function is $0$ everywhere, except that it is $1$ at $\{1,2,3\}$. We therefore find that \[E_0=E_1=E_2 = 1 \Leftrightarrow \Delta_0=\Delta_1=\Delta_2=0\] and \[E_3=0 \Leftrightarrow \Delta_3=1.\] We have seen that $d^*_1(C') = 3$, so that for $\mu<3$, the Theorem gives $\Delta_\mu=0$, while it gives $\Delta_3=1$.
\end{example}

\begin{example} Let $q$ be a prime power, $k \leqslant q-1$ and let $r\geqslant 2$ be  such that $r \mid q-1$ and $r \mid k$. Let $\gamma \in \Fq^*$ be a generator of $\Fq^*$. A generator matrix of the Reed-Solomon code $RS_{q,\gamma,k} \subset \Fq^{q-1}$ is given by \[G=\begin{bmatrix} 1 & 1 & \hdots & 1 \\
\gamma & \gamma^2 & \hdots & \gamma^{q-1} \\
\gamma^2 & \gamma^4 & \hdots & \gamma^{2(q-1)} \\
\vdots & \vdots & \ddots & \vdots \\
\gamma^{k-1} & \gamma^{2(k-1)} & \hdots & \gamma^{(k-1)(q-1)}\end{bmatrix}.\] We consider the $r$-folded Reed-Solomon code $FRS_{q,\gamma,r,k}$ defined in the following way (see~\cite{GR}): let $\phi$ be \[\begin{array}{ccc}  \Fq^{q-1} & \longrightarrow & \left(\Fq^r\right)^\frac{q-1}{r} 
\\ (x_1\cdots,x_{q-1}) & \longmapsto &\left((x_1,\cdots,x_r),(x_{r+1},\cdots,x_{2r}),\cdots\right)
\end{array}.\] Then \[FRS_{q,\gamma,r,k} = \phi(RS_{q,\gamma,k}).\]  This is a block code of length $\frac{q-1}{r}$ on the alphabet $\Fq^r$.

We use the notation of section~\ref{five}. If $X \subset \{1,\cdots,\frac{q-1}{r}\}$, then the submatrix $G_{X_r}$ is a Vandermonde matrix, and as such, we have \[\rk_{\Fq} G_{X_r} = \min \{|X_r|,k\},\] which is obviously divisible by $r$. This shows that the $r$-folded Reed-Solomon code is a multilinear code over $\Fq^r.$  Now, \[\rk_{\Fq} G_{X_r} = \min \{|X_r|,k\}\] which implies\[ \left|{FRS_{q,\gamma,r,k}}_X\right| = q^{\rk_{\Fq} G_{X_r}}= \left\{\begin{array}{ll}  \left(q^r\right)^{|X|} & \textrm{ if } |X|\leqslant \frac{k}{r} \\ \left(q^r\right)^\frac{k}{r} & \textrm{ if }|X| >k\end{array}\right.\] This shows that the matroid associated to the $r$-folded Reed-Solomon code is the uniform matroid  $U_{\frac{k}{r},\frac{q-1}{r}}$ on $\frac{q-1}{r}$ elements and rank $\frac{k}{r}$, and its generalized Hamming weights are \[d_i(FRS_{q,\gamma,r,k}) = \frac{q-1-k}{r} + i\] for $1 \leqslant i \leqslant \frac{k}{r}$ and \[d_i(FRS_{q,\gamma,r,k})^* = \frac{k}{r} + i\] for $1 \leqslant i \leqslant \frac{q-1-k}{r}$. It is therefore an MDS-code. 

Let  $\varphi: \left(\Fq^r\right) ^\frac{q-1-k}{r} \times \left(\Fq^r\right) ^\frac{q-1-k}{r} \rightarrow \left(\Fq^r\right) ^\frac{q-1-k}{r}$ is an application as described above, for example componentwise addition. By the above description of the generalized Hamming weights, an intruder does not get any digit of information if he is able to listen up to $\frac{k}{r}-1$ digits of the sent message, he gets $i$ digits of information if he is able to listen to $\frac{k}{r}+i-1$ digits of the sent message.

If we want to keep the same robustness again intruders with a linear code on a field with the same alphabet size, we have to use an MDS-code over $\mathbb{F}_{q^r}$ (for example a punctured Reed-Solomon code of dimension $\frac{k}{r}$ where we only keep $\frac{q-1}{r}$ columns of a generator matrix). It is easy to see that it gives the same robustness than the scheme presented above, since both are MDS. The benefit of using a folded Reed-Solomon code is that the computations are done over the smaller field $\Fq$ instead of $\mathbb{F}_{q^r}$.

\end{example}

\begin{acknowledgements}

The authors would like to thank IMPA, Rio de Janeiro, where a part of the first named author's work with this article was done, during the special trimester April-June 2015.

The authors would also like to thank the anonymous referee for a series of comments that led to a significant improvement of the article.

\end{acknowledgements}

\end{document}